\documentclass[9pt]{article}
\usepackage{spconf,amsmath,epsfig,color,rrprefs}
\usepackage{amssymb,epsfig,graphicx,subfigure}

\title{SAMPLE SIZE COGNIZANT DETECTION OF SIGNALS IN WHITE NOISE}
\twoauthors
  {Raj Rao Nadakuditi\sthanks{Thanks NSF  DMS-0411962 and ONR N00014-07-1-0269.}}
   {Massachusetts Institute of Technology\\
   Department of EECS\\
   Cambridge, MA 02139.}
  {Alan Edelman}
   {Massachusetts Institute of Technology\\
   Department of Mathematics\\
   Cambridge, MA 02139.}
\begin{document}
%
\maketitle
\begin{abstract}
The detection and estimation of signals in noisy, limited data is a problem of interest to many scientific and engineering communities. We present a computationally simple, sample eigenvalue based procedure for estimating the number of high-dimensional signals in white noise when there are relatively few samples. We highlight a fundamental asymptotic limit of sample eigenvalue based detection of weak high-dimensional signals from a limited sample size and discuss its implication for the detection of two closely spaced signals. 

This motivates our heuristic definition of the \textit{effective number of identifiable signals}. Numerical simulations are used to demonstrate the consistency of the algorithm with respect to the effective number of signals and the superior performance of the algorithm with respect to Wax and Kailath's ``asymptotically consistent'' MDL based estimator. 

\end{abstract}
\section{INTRODUCTION}
\label{sec:intro}
The observation vector, in many signal processing applications, can be modelled as a superposition of a finite number of signals embedded in additive noise. Detecting the number of signals present becomes a key issue and is often the starting point for the signal parameter estimation problem. When the signals and the noise are assumed to be samples of a stationary, ergodic Gaussian vector process, the sample covariance matrix formed from $m$ observations has the Wishart distribution. 

The proposed algorithm uses an information theoretic criterion, motivated by the approach taken by Wax and Kailath (henceforth WK) in \cite{kailath-wax}, for determining the number of signals in white noise by performing inference on the eigenvalues of the resulting sample covariance matrix. The form of the estimator is motivated by the distributional properties of moments of the eigenvalues of large dimensional Wishart matrices \cite{dumitriu05a}. 

The proposed estimator was derived by explicitly accounting for the blurring and fluctuations of the eigenvalues due to sample size constraints. Consequently, there is a greater theoretical justification for employing the proposed estimator in sample starved settings unlike the WK estimators which were derived assuming that the sample size greatly exceeds the number of sensors. This is reflected in the improved performance relative to the ``asymptotically consistent'' WK MDL based estimator. 

Another important contribution of this paper is the description of a fundamental limit of eigen-inference, \textit{i.e.}, inference using the sample eigenvalues alone. The concept of \textit{effective number of identifiable signals}, introduced herein, explains why, asymptotically, if the signal level is below a threshold that depends on the noise variance, sample size and the dimensionality of the system, then reliable detection is not possible. 

This paper is organized as follows. The problem is formulated in Section \ref{sec:lrcf problem formulation}. An estimator for the number of signals present that exploits results from random matrix theory is derived in Section \ref{sec:number of signals}. The fundamental limits of sample eigenvalue based detection and the concept of \textit{effective number of signals} are discussed in Section \ref{sec:effective number of signals}. Simulation results are presented in Section \ref{sec:lrcf simulations} while some concluding remarks and directions for future research are presented in Section \ref{sec:lrcf conclusion}.

\section{PROBLEM FORMULATION}
\label{sec:lrcf problem formulation}

We observe $m$ samples (``snapshots'') of possibly signal bearing $n$-dimensional snapshot vectors ${\bf x}_{1}, \ldots, {\bf x}_{m}$ where for each $i$, ${\bf x}_{i} \sim \mathcal{N}_{n}(0,{\bf R})$ and ${\bf x}_{i}$ are mutually independent.  The snapshot vectors are modelled as 
\begin{equation}\label{eq:superposition problem}
{\bf x}_{i} = {\bf A}\,{\bf s}_{i}+{\bf z}_{i}  \qquad \textrm{for } i = 1,\ldots,m,
\end{equation}
where ${\bf z}_{i} \sim \mathcal{N}_{n}(0,\sigma^{2}{\bf I})$, denotes an $n$-dimensional (real or complex) Gaussian noise vector where $\sigma^{2}$ is generically unknown, ${\bf s}_{i} \sim \mathcal{N}_{k}({\bf 0},{\bf R}_{s})$ denotes a $k$-dimensional (real or complex) Gaussian signal vector with covariance ${\bf R}_{s}$, and ${\bf A}$ is a $n \times k$  unknown non-random matrix. 

Since the signal and noise vectors are independent of each other, the covariance matrix of ${\bf x}_{i}$ can hence be decomposed as
\begin{equation}\label{eq:R model}
{\bf R} =  {\bf \Psi} + \sigma^{2} {\bf I}
\end{equation}
where
\begin{equation}
{\bf \Psi} = {\bf A}{\bf R}_{s}{\bf A}' ,
\end{equation}
with $'$ denoting the conjugate transpose. Assuming that the matrix ${\bf A}$ is of full column rank, {\it i.e.}, the columns of ${\bf A}$ are linearly independent, and that the covariance matrix of the signals ${\bf R}_{s}$ is nonsingular, it follows that the rank of ${\bf \Psi}$ is $k$. Equivalently, the $n-k$ smallest eigenvalues of ${\bf \Psi}$ are equal to zero. 

If we denote the eigenvalues of ${\bf R}$ by $\lambda_{1}\geq \lambda_{2} \geq \ldots \geq \lambda_{n}$ then it follows that the smallest $n-k$ eigenvalues of ${\bf R}$ are all equal to $\sigma^{2}$ so that
\begin{equation}
\lambda_{k+1}= \lambda_{k+2} = \ldots = \lambda_{n} = \lambda = \sigma^{2}.
\end{equation}
Thus, if the true covariance matrix ${\bf R}$ were known apriori, the dimension of the signal vector $k$ can be determined from the multiplicity of the smallest eigenvalue of ${\bf R}$.  The problem in practice is that the covariance matrix ${\bf R}$ is unknown so that such a straight-forward algorithm cannot be used. The signal detection and estimation problem is hence posed in terms of an inference problem on $m$ samples of $n$-dimensional multivariate real or complex Gaussian snapshot vectors.  

A classical approach to this problem, developed by Bartlett \cite{bartlett54a} and Lawley \cite{lawley56a}, uses a sequence of hypothesis tests. Though this approach is sophisticated, the main problem is the subjective judgement needed by the practitioner in selecting the threshold levels for the different tests. This was overcome by Wax and Kailath  in \cite{kailath-wax} wherein they propose an estimator for the number of signals (\underline{assuming $m>n$}) based on the eigenvalues $l_{1} \geq l_{2} \geq \ldots \geq l_{n}$ of the sample covariance matrix (SCM) defined by
\begin{equation}
\widehat{{\bf R}} = \frac{1}{m} \sum_{i=1}^{m} {\bf x}_{i} {\bf x}_{i}' = \frac{1}{m} {\bf X}{\bf X}'
\end{equation}
where ${\bf X}= [{\bf x}_{1}| \ldots| {\bf x}_{m}]$ is the matrix of observations (samples).  The Akaike Information Criteria (AIC) form of the estimator is given by
\begin{equation}\label{eq:aic est}
\hat{k}_{{\rm AIC}} = \argmin_{k \in \mathbb{N}: 0 \leq k < n}  -2(n-k)m \log \frac{g(k)}{a(k)}+ 2k(2n-k)
\end{equation}
while the Minimum Descriptive Length (MDL) criterion is given by
\begin{multline}\label{eq:mdl est}
\hat{k}_{{\rm MDL}} = \argmin_{k \in \mathbb{N}: 0 \leq k < n} - (n-k)m\log \frac{g(k)}{a(k)} \\+ \frac{1}{2}k(2n-k)\log m
\end{multline}
where $g(k) = \prod_{j=k+1}^{n} l_{j}^{1/(n-k)}$ is the geometric mean of the $n-k$ smallest sample eigenvalues and $a(k)= \frac{1}{n-k} \sum_{j = k+1}^{n} l_{j}$ is their arithmetic mean. 

It is known \cite{kailath-wax} that the AIC form inconsistently estimates the number of signals, while the MDL form estimates the number of signals consistently. The simplicity of the estimator, and the large sample consistency are among the primary reasons why the Kailath-Wax MDL estimator continues to be employed in practice \cite{vantrees02a}. In the two decades since the publication of the WK paper, researchers have come up with many innovative solutions  (\cite{fishler00a, fishler05a, larocque02a} to list a few) for making the estimators more robust by exploiting some type of prior knowledge. 

The most important deficiency of the WK and related estimators that remains unresolved occurs when the sample size is smaller than the number of sensors, \ie, when $m < n$. In this situation, the SCM is singular and the estimators become degenerate. Practitioners often overcome this in an ad-hoc fashion by, for example, restricting $k$ in (\ref{eq:mdl est}) to integer values in the range $0 \leq k < \min(n,m)$. Since large sample, \ie, $m \gg n$, asymptotics were used to derive the estimators in \cite{kailath-wax}, there is no rigorous theoretical justification for such a reformulation even if the simulation results suggest that the WK estimators are working ``well enough.''

Other sample eigenvalue based solutions found in the literature that exploit the sample eigenvalue order statistics \cite{fishler00a}, or employ a Bayesian framework by imposing priors on the number of signals \cite{bansal91a} are computationally more intensive and do not address the sample size starved setting in their analysis or their simulations. Particle filter based techniques \cite{larocque02a}, while useful, require the practitioner to the model the eigenvectors of the underlying population covariance matrix as well; this makes them especially sensitive to model mismatch errors that are endemic to high-dimensional settings. This motivates our development of an sample eigenvalue based estimator with a computational complexity comparable to that of the WK estimators.

\section{ESTIMATING THE NUMBER OF SIGNALS}\label{sec:number of signals}
Given an observation ${\bf y} = [y(1), \ldots, y(N)]$ and a family of models, or equivalently a parameterized family of probability densities $f({\bf y}|\bm{\theta})$ indexed by the parameter vector $\bm{\theta}$, we select the model which gives the minimum Akaike Information Criterion (AIC) \cite{akaike74a} defined by
\begin{equation}\label{eq:aic criterion}
{\rm AIC}_{k} = -2 \log f({\bf y} | \widehat{\bm{\theta}}) + 2 k  
\end{equation}
where $\widehat{\bm{\theta}}$ is the maximum likelihood estimate of $\bm{\theta}$, and $k$ is the number of free parameters in $\bm{\theta}$. We derive an AIC based estimator for the number of signals by exploiting the following distributional properties of the moments of eigenvalues of the (signal-free) SCM.

\begin{theorem}(Dumitriu-Edelman \cite{dumitriu05a})\label{th:fluctuations}
Assume $\widehat{{\bf R}}$ is formed from $m$ snapshots modelled as (\ref{eq:superposition problem}) with $k=0$, $\lambda=1$ then as $m,n \to \infty$ and $c_{m} = n/m \to c \in (0,\infty)$, then 
\begin{equation*}
\begin{bmatrix}
\sum_{i=1}^{n} l_{i} - n \\
\\
\sum_{i=1}^{n} l_{i}^{2} - n\,(1+c) - (\frac{2}{\beta}-1)c
\end{bmatrix}
\overset{\mathcal{D}}{\to} \mathcal{N}\left({\bf 0},{\bf Q}\right)
\end{equation*}
where $\mathcal{D}$ denotes convergence in distribution, $\beta = 1$ (or $2$) when ${\bf x}_{i}$ is real (or complex) valued, respectively, and
\begin{equation*}
{\bf Q} = 
\frac{2}{\beta}\begin{bmatrix}
c & 2c\,(c+1) \\
2c\,(c+1)  & 2c\,(2c^2+5c+2)\\
\end{bmatrix}.
\end{equation*}
\end{theorem}

\begin{proposition}\label{prop:main prop}
Assume $\widehat{{\bf R}}$ satisfies the hypotheses of Theorem \ref{th:fluctuations} for some $\lambda$ then as $m,n \to \infty$ and $c_{m} = n/m \to c \in (0,\infty)$, then
\begin{equation}
n\left[t_{n}-(1+c)\right] \overset{\mathcal{D}}{\to} \mathcal{N}\left(\left(\frac{2}{\beta}-1\right)c,\frac{4}{\beta}c^{2}\right) 
\end{equation}
and the test statistic $t_{n}$ is given by
\[
t_{n} = \frac{\tfrac{1}{n}\sum_{i} l_{i}^2}{\left(\tfrac{1}{n}\sum_{i} l_{i}\right)^{2}} = \frac{{\rm Second~moment~of~eigs}}{{\rm Mean~sq.~of~eigs}}
\]
\end{proposition}
\begin{proof}
This follows from applying the delta method \cite{casella90a} to the results in Theorem \ref{th:fluctuations}.
\end{proof}

When $k>0$ signals are present and assuming $k \ll n$, then the distributional properties of the $n-k$ ``noise'' eigenvalues are closely approximated by the distributional properties of the eigenvalues given by Theorem \ref{th:fluctuations} of the signal-free SCM, {\it i.e.}, $k=0$. Hence, by evaluating the statistic in Proposition \ref{prop:main prop} over a sliding window, and using the normal  approximation for the statistic from Proposition \ref{prop:main prop} with $c \approx n/m$ and $k+1$ free parameters in the AIC formulation in (\ref{eq:aic criterion}) results in the estimator:

\begin{subequations} \label{eq:THE estimator}
\begin{empheq}[
box=\setlength{\fboxrule}{1pt}\fbox]{equation}
\hat{k}_{{\rm NEW}}= \!\argmin_{k \in \mathbb{N}: 0\leq k < \min(n,m)} \! \left\{\frac{\beta}{4} \left[\frac{m}{n}\right]^{2} \! q_{k}^{2} \right\}\! + \! 2(k+1), 
\end{empheq}
\textrm{where}
\begin{multline}
q_{k} = n\left[\underbrace{\frac{\tfrac{1}{n-k}\sum_{i=k+1}^{n} l_{i}^{2}}{(\tfrac{1}{n-k}\sum_{i=k+1}^{n} l_{i})^{2}}}_{t_{n,k}} - \left(1+\frac{n}{m}\right) \right]-\\
\left(\frac{2}{\beta}-1\right)\frac{n}{m}.
\end{multline}
\end{subequations}

Here $\beta = 1$ if ${\bf x}_{i} \in \mathbb{R}^{n}$, and $\beta = 2$ if ${\bf x}_{i} \in \mathbb{C}^{n}$. When the measurement vectors represent quaternion valued narrowband signals, then we set $\beta =4$. Quaternion valued vectors arise when the data collected from vector sensors is represented using quaternions as in \cite{miron06}.

\section{FUNDAMENTAL LIMIT OF DETECTION}\label{sec:effective number of signals}
The following result exposes when the ``signal'' eigenvalues are asymptotically distinguishable from the ``noise'' eigenvalues.
\begin{proposition}\label{prop:spiked convergence}
Assume $\widehat{{\bf R}}$ satisfies the hypotheses of Theorem \ref{th:fluctuations}. Denote the eigenvalues of ${\bf R}$ by $\lambda_{1} \geq \lambda_{2} > \ldots \geq \lambda_{k} > \lambda_{k+1} = \ldots \lambda_{n} = \lambda= \sigma^{2}$. Let $l_{j}$ denote the $j$-th largest eigenvalue of $\widehat{{\bf R}}$. Then as $n,m \to \infty$ with $c_{m} = n/m \to c \in (0,\infty)$, and $j=1,\ldots,k+1$,
\begin{equation}
l_{j} \to 
\begin{cases}
\lambda_{j} \left( 1+ \dfrac{\sigma^{2}\,c}{\lambda_{j}-\sigma^{2}}\right) & {\rm if } \, \lambda_{j} > \sigma^{2}\,(1+\sqrt{c})\\
& \\
\sigma^{2}\,(1+\sqrt{c})^{2} & {\rm if } \, \lambda_{j} \leq \sigma^{2}(1+\sqrt{c})\\
\end{cases}
\end{equation}
where the convergence is almost surely.
\end{proposition}
\begin{proof}
This result appears in \cite{BaikS06} for very general settings. A matrix theoretic proof for when $c<1$ for the real case may be found in \cite{Paul05a} and an interacting particle system interpretation appears in \cite{raj:thesis}.
\end{proof}

Motivated by Proposition \ref{prop:spiked convergence}, we define the effective number of signals as
\begin{empheq}[
box=\setlength{\fboxrule}{1pt}\fbox]{equation}\label{eq:effective rank}
k_{{\rm eff}}({\bf R}) = \#\textrm{ eigs. of }{\bf R} > \sigma^{2}\left(1+\sqrt{\dfrac{n}{m}}\right).
\end{empheq}

\subsection{Identifiability of closely spaced signals}
Suppose there are two uncorrelated (hence, independent) signals so that ${\bf R}_{s} = \textrm{diag}(\sigma_{{\rm S}1}^{2},\sigma_{{\rm S}2}^{2})$. In (\ref{eq:superposition problem}), let ${\bf A} = [{\bf v}_{1} {\bf v}_{2}]$. In a sensor array processing application, we think of ${\bf v}_{1} \equiv {\bf v}(\theta_{1})$ and ${\bf v}_{2} \equiv {\bf v}_{2}(\theta_{2})$ as encoding the array manifold vectors for a source and an interferer with powers $\sigma_{{\rm S}1}^{2}$ and $\sigma_{{\rm S}2}^{2}$, located at $\theta_{1}$ and $\theta_{2}$, respectively. The covariance matrix given by 
\begin{equation}
{\bf R} = \sigma_{{\rm S}1}^{2} {\bf v}_{1}{\bf v}_{1}'+ \sigma_{{\rm S}2}^{2}  {\bf v}_{2}{\bf v}_{2}' + \sigma^{2} {\bf I}
\end{equation}
has the $n-2$ smallest eigenvalues $\lambda_{3} = \ldots = \lambda_{n} = \sigma^{2}$ and the two largest eigenvalues
\begin{subequations}\label{eq:ev 2 sources}
\begin{multline}
\lambda_{1} =  
\sigma^{2}+ \tfrac{\left(\sigma_{{\rm S}1}^{2} \parallel \! {\bf v}_{1} \!\parallel^{2}+\sigma_{{\rm S}2}^{2} \parallel \! {\bf v}_{2} \!\parallel^{2}\right)}{2} 
\\+ \tfrac{
\sqrt{\left(\sigma_{{\rm S}1}^{2} \parallel \! {\bf v}_{1} \!\parallel^{2}-\sigma_{{\rm S}2}^{2} \parallel \! {\bf v}_{2} \!\parallel^{2}\right)^{2}+4\sigma_{{\rm S}1}^{2}\sigma_{{\rm S}2}^{2} |\langle {\bf v}_{1}, {\bf v}_{2} \rangle| ^{2}}}{2}
\end{multline}
\begin{multline}
\lambda_{2} =  \sigma^{2}+ \tfrac{\left(\sigma_{{\rm S}1}^{2} \parallel \! {\bf v}_{1} \!\parallel^{2}+\sigma_{{\rm S}2}^{2} \parallel \! {\bf v}_{2} \!\parallel^{2}\right)}{2} \\
-\tfrac{\sqrt{\left(\sigma_{{\rm S}1}^{2} \parallel \! {\bf v}_{1} \!\parallel^{2}-\sigma_{{\rm S}2}^{2} \parallel \! {\bf v}_{2} \!\parallel^{2}\right)^{2}+4\sigma_{{\rm S}1}^{2}\sigma_{{\rm S}2}^{2} |\langle {\bf v}_{1}, {\bf v}_{2} \rangle| ^{2}}}{2}
\end{multline}
\end{subequations}
respectively. Applying the result in Proposition \ref{prop:spiked convergence} allows us to express the effective number of signals as  
\begin{equation}
k_{{\rm eff}} = 
\begin{cases}
2 &\qquad\textrm{if    } \phantom{~~~~}\sigma^{2} \left(1+\sqrt{\dfrac{n}{m}} \right) < \lambda_{2}\\
  &\\
1 &\qquad\textrm{if    } \phantom{~~~~}\lambda_{2} \leq \sigma^{2} \left(1+\sqrt{\dfrac{n}{m}} \right) < \lambda_{1}\\
   &\\
0 &\qquad\textrm{if    } \phantom{~~~~}\lambda_{1} \leq \sigma^{2} \left(1+\sqrt{\dfrac{n}{m}} \right) \\
\end{cases}
\end{equation}
In the special situation when  $\parallel\! {\bf v}_{1} \!\parallel = \parallel\! {\bf v}_{2}\! \parallel = \parallel\! {\bf v}\! \parallel$ and $\sigma_{{\rm S1}}^{2} = \sigma_{{\rm S2}}^{2} = \sigma_{{\rm S}}^{2}$, we can (in an asymptotic sense) reliably detect the presence of \textit{both signals} from the sample eigenvalues alone whenever
\begin{empheq}[
box=\setlength{\fboxrule}{1pt}\fbox]{equation}
\label{eq:array tradeoff}
\sigma_{{\rm S}}^{2} \parallel \! {\bf v} \!\parallel^{2} \left(1-\dfrac{|\langle {\bf v}_{1},{\bf v}_{2}\rangle |}{\parallel \! {\bf v}\parallel}\right)  > \sigma^{2} \sqrt{\dfrac{n}{m}}  
\end{empheq}
Equation (\ref{eq:array tradeoff}) captures the tradeoff between the identifiability of two closely spaced signals, the dimensionality of the system, the number of available snapshots and the cosine of the angle between the vectors ${\bf v}_{1}$ and ${\bf v}_{2}$. It may prove to be a useful heuristic for experimental design.

\section{SIMULATIONS}
\label{sec:lrcf simulations}

\begin{figure}[t]
\subfigure[Prob$(\hat{k} = 2)$ versus $n$ for fixed $n/m$.]
{
\label{fig:numerics 1}
\includegraphics[width=7.9cm]{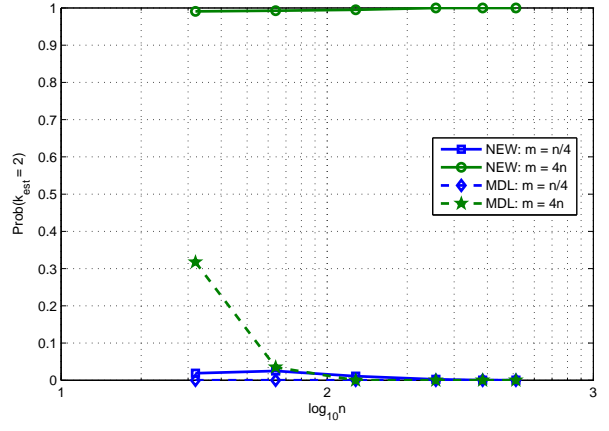}
}
\subfigure[Prob$(\hat{k} = 1)$ versus $n$ for fixed $n/m$.]
{
\label{fig:numerics 2}
\includegraphics[width=7.9cm]{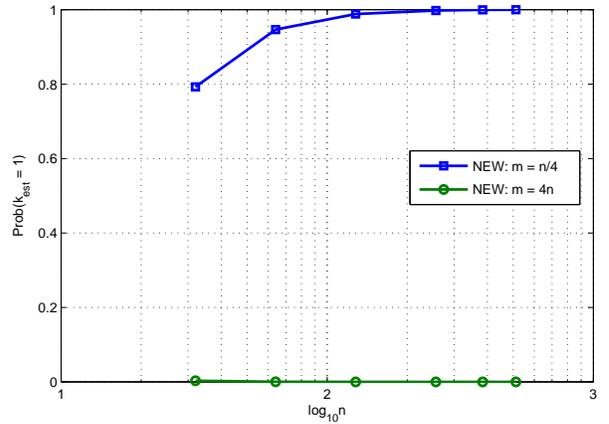}
}
\caption{Comparison of the estimators over 20,000 trials.}
\label{fig:numerics}
\end{figure}

Assume the covariance matrix ${\bf R}$ has $n-2$ ``noise'' eigenvalues with $\sigma^{2}=1$, and two ``signal'' eigenvalues with $\lambda_{1} = 10$ and $\lambda_{2} = 3$. When $m = 4n$ samples are available, Figure \ref{fig:numerics 1} shows that the proposed estimator consistently detects two signals while the WK MDL estimator does not. However, when $m = n/4$, Figure \ref{fig:numerics 1} suggests that neither estimator is able to detect both the signals present. A closer examination of the empirical data presents a different picture. For the covariance matrix considered, when $m = n/4$, then from (\ref{eq:effective rank}), $k_{eff}=1$. Figure \ref{fig:numerics 2} shows that for large $n$ and  $m = n/4$, the new estimator consistently estimates one signal, as expected. The WK MDL estimator detects no signals. We conjecture that the new estimator consistently estimates $k_{eff}$ in the $n,m \to \infty, n/m \to c$ sense.

\section{CONCLUSIONS}\label{sec:lrcf conclusion}

An estimator for the number of signals in white noise was presented that exhibits robustness to high-dimensionality, and sample size constraints. The concept of \textit{effective number of signals} described provides insight into the (asymptotic) regime in which reliable detection with sample eigenvalue based methods, including the proposed method, is possible. This helps identify scenarios where algorithms that exploit any structure in the eigenvectors of the signals, such as the MUSIC and the Capon-MVDR \cite{vantrees02a} algorithms in sensor array processing, might be better able to tease out lower level signals from the background noise. It is worth noting that the proposed approach remains relevant in situations where the eigenvector structure has been identified. This is because eigen-inference methodologies are inherently robust to eigenvector modelling errors that are endemic to high-dimensional settings. Thus the practitioner may use the proposed estimator to complement and ``robustify'' the inference provided by  algorithms that exploit the eigenvector structure. 

\begin{center}
{\bf ACKNOWLEDGEMENTS}\\
\end{center}

We thank Arthur Baggeroer, William Ballance and the anonymous reviewers for their feedback and encouragement.


\bibliographystyle{IEEEbib}
\bibliography{randbib}

\end{document}